%% file: main.tex
\documentclass[12pt]{article}
\usepackage[english]{babel}
\usepackage{bm}
\usepackage{fullpage}
\usepackage{cite}

\input{macros}

\usepackage{authblk}
\newcommand{\Ran}{{\rm{Ran}}\,}

\newtheorem*{remark}{Remark}

\allowdisplaybreaks[1] 
\usepackage{graphicx,import,xcolor,transparent}
\usepackage{centernot}

\begin{document}

\title{Entropy Concentration and Universal Typicality for Almost i.i.d.\ Quantum Sources}
\author{Nilanjana Datta}

\affil{University of Cambridge\\
Department of Applied Mathematics and Theoretical Physics\\
University of Cambridge\\
Cambridge CB3 0WA, United Kingdom
}

\date{}

\maketitle

\begin{abstract}
We establish two concentration principles for almost i.i.d.\ quantum sources under the weakly almost i.i.d.\ condition, the broadest condition in the established hierarchy of almost i.i.d.\ sources. Our results therefore apply also to Wasserstein- and Mazzola-Sutter-Renner almost i.i.d.\ sources. Weakly almost i.i.d.\ sources are sequences of multipartite states whose fixed-size marginals converge, on average, to tensor powers of a reference state, while allowing arbitrary global correlations and entanglement. We prove a noncommutative weak law of large numbers for empirical observables, and a universal entropy-concentration principle showing asymptotic concentration on subspaces of exponential dimension governed by the von Neumann entropy of the reference state.
{These concentration principles provide a unified and conceptually
transparent approach to several information-theoretic applications
beyond the i.i.d.\ setting, including direct proofs of universal
compression within classes of weakly almost i.i.d.\ sources sharing
a common reference state and asymmetric quantum hypothesis-testing
bounds. They also yield concentration of empirical local-measurement
frequencies and of empirical averages of commuting single-site
observables, with applications to single-site generalized Gibbs
ensembles, together with bounds on smooth and spectral entropies.}
\end{abstract}

{\section{Introduction}
{Classical and quantum Shannon theories are traditionally formulated under the assumption that resources are independent and identically distributed (i.i.d.) across many uses. However, exact i.i.d.\ structure is often unrealistic in the presence of correlations, adversarial behaviour, or experimental imperfections, and may in practice be impossible to verify operationally. This has motivated several approaches to information theory beyond the i.i.d.\ paradigm, most notably the information-spectrum framework introduced by Han and Verdú~\cite{HanVerdu1993} and the smooth-entropy formalism~\cite{RennerBook, TOMAMICHEL} of one-shot information theory, both of which avoid exact tensor-power assumptions altogether; their relation is discussed in~\cite{DattaRenner2009}).

A complementary line of work investigates situations in which the i.i.d.\ assumption fails only mildly, with the aim of understanding which information-theoretic properties survive under controlled correlations or defects. This has led to several notions of ``almost i.i.d.'' behaviour. Early progress in this direction was made by Renner, who introduced the notion of ``almost power'' states in quantum cryptography~\cite{Renner2007}. Closely related ideas later appeared in the work of Brand{\~a}o and Plenio on the generalized quantum Stein’s lemma~\cite{BrandaoPlenio}, and more recently in the work of Lami~\cite{Lami_2025}. More recent works, including~\cite{mazzolasutterrenner1, girardi2026newapproachesiidinformation}, have compared different notions of approximate i.i.d.\ structure and clarified the relations between them.
These works make clear that ``almost i.i.d.'' is not a single well-defined notion. Different definitions capture different operational and physical aspects of approximate independence, and they form a genuine hierarchy of conditions. Consequently, understanding robustness beyond the i.i.d.\ regime requires determining which operational properties persist under which forms of approximate independence (see e.g.~\cite{girardi2026quantumshannontheoryrobust} and~\cite{mazzola2026robustgeneralizedquantumsteins}).

{In this paper we establish concentration principles for almost i.i.d.\ quantum sources by working with the broadest class in this hierarchy: weakly almost i.i.d.\ sources~\cite{girardi2026newapproachesiidinformation}. For notational simplicity we henceforth use the acronym {\em{iid}} for i.i.d. Every Mazzola--Sutter--Renner (MSR) almost iid source~\cite{mazzolasutterrenner1} is Wasserstein almost iid~\cite{girardi2026newapproachesiidinformation}, and every Wasserstein almost iid source is weakly almost iid, along the same reference state (see Fig 1 of~\cite{girardi2026newapproachesiidinformation}). Consequently, the concentration principles of this paper, and the resulting achievability and entropy bounds, apply throughout this hierarchy. In particular, their applicability extends to all three notions of almost iid sources.}

{The weakly almost iid condition is an asymptotic condition
on fixed-size marginals.
{Roughly speaking, a sequence $\bm\rho = (\rho_n)_n$ of quantum states
is weakly almost iid along a reference state $\rho$ if, for every
fixed $k$, the trace-norm distances
$\|(\rho_n)_I - \rho^{\otimes k}\|_1$ tend to zero on average over
all $k$-element subsets $I \subseteq [n]$ as $n \to \infty$.}
This condition does not require permutation invariance or global trace-norm convergence to $\rho^{\otimes n}$. It permits strong global correlations and multipartite entanglement. For example, weakly-almost iid sources along the maximally mixed state may consist of globally pure, highly entangled states.}

The main contribution of this paper is the identification of two concentration principles for weakly almost iid sources. The first is a noncommutative weak law of large numbers, showing spectral concentration of empirical averages of local observables around their expectations with respect to the reference state. The second is a universal entropy-concentration principle, establishing asymptotic concentration on subspaces whose exponential dimension is governed by the von Neumann entropy of the reference state.

{These principles provide a unified and conceptually transparent framework for several applications beyond the iid setting. These include direct proofs of universal quantum data compression within classes of weakly almost iid sources sharing a common reference state, asymmetric quantum hypothesis-testing bounds, and bounds on smooth and spectral entropy quantities. In several cases, the resulting proofs are significantly more direct than previous approaches~\cite{girardi2026quantumshannontheoryrobust}, with the relevant operational statements emerging naturally from the underlying concentration principles. We also derive concentration results for macroscopic observables in quantum many-body systems, including generalized Gibbs ensembles, and for repeated local measurement statistics.}

A central message of this work is that the asymptotic equipartition property can be substantially weakened. While weakly almost i.i.d. sources generally do not satisfy any approximate equipartition property—indeed, they may exhibit arbitrary long-range correlations and need not be approximately maximally mixed on their entropy-typical subspaces—they nevertheless exhibit entropy concentration onto subspaces of dimension $e^{nS(\rho)+o(n)}$. Remarkably, this weaker form of typicality already suffices for a number of important information-theoretic applications, including universal compression, hypothesis testing, and applications to quantum many-body systems. 
Our results suggest that entropy concentration, rather than equipartition itself, is the essential mechanism underlying a variety of operational and thermodynamic consequences of typicality.

\medskip

\noindent
{\bf Layout of the paper:}
Section~\ref{sec:walmost iid} introduces the necessary notation and recalls the definition of weakly almost iid quantum sources. Section~\ref{sec:key-lemmas} establishes the two concentration principles underlying the paper: a noncommutative weak law of large numbers and a universal entropy-concentration theorem. Section~\ref{sec:applns} contains applications. Section~\ref{sec:data-comp} studies universal quantum data compression, Section~\ref{sec:HT} asymmetric quantum hypothesis testing, Section~\ref{sec:many-body} concentration of macroscopic observables in many-body systems and generalized Gibbs ensembles, and Section~\ref{sec:mmt} concentration of repeated local measurement statistics. Finally, Sections~\ref{sec:smooth} and~\ref{sec:info-spec} derive consequences for smooth and spectral entropy quantities.}

\medskip

{
\section{Weakly almost i.i.d.~sources}
\label{sec:walmost iid}
{Let $\mathcal H$ be a finite-dimensional Hilbert space. Let $\mathcal L(\mathcal H)$ denote the algebra of linear operators acting on $\mathcal H$, and let $\mathcal D(\mathcal H)\subset\mathcal L(\mathcal H)$ denote the set of quantum states (or density operators), i.e.\ positive semidefinite operators of unit trace on $\mathcal H$. A \emph{pure state} is a rank-one projection of the form $\ket{\psi}\!\bra{\psi}$, where $\ket{\psi}\in\mathcal H$ is a unit vector.

The identity operator on $\mathcal H$ is denoted by $\id$, while $\id_n$ denotes the identity operator of ${\mathcal H}^{\otimes n}$. An operator $\Pi\in\mathcal L(\mathcal H)$ is called an \emph{orthogonal projection} if $\Pi=\Pi^\dagger=\Pi^2$. The range of $\Pi$ is denoted by
$
\Ran(\Pi):=\{\Pi\ket{\psi}:\ket{\psi}\in\mathcal H\}.
$
For a positive semidefinite operator $X\in\mathcal L(\mathcal H)$, its support is defined as
$
\supp(X):=\Ran(\{X>0\}),
$
where $\{X>0\}$ denotes the orthogonal projection onto the span of all eigenvectors of $X$ corresponding to strictly positive eigenvalues. For an operator $A \in \mathcal L(\mathcal H)$, we denote by $A^{(i)}$ the operator in $\mathcal L(\mathcal H^{\otimes n})$ given by $A^{(i)}:=\id^{\otimes(i-1)}\otimes A\otimes \id^{\otimes(n-i)}$. 

For \(p\in[1,\infty)\), the Schatten \(p\)-norm of an operator \(X\in\mathcal L(\mathcal H)\) is defined by
\(
\|X\|_p:=\bigl(\Tr|X|^p\bigr)^{1/p},
\)
where
\(
|X|:=\sqrt{X^\dagger X}.
\)
The Schatten \(\infty\)-norm is the operator norm
\(
\|X\|_\infty:=\sup_{\|\psi\|=1}\|X\psi\|.
\)
In particular,
\(
\|X\|_1=\Tr|X|
\)
is the trace norm, and
\(
\|X\|_2=\sqrt{\Tr(X^\dagger X)}
\)
is the Hilbert--Schmidt norm.
For a self-adjoint operator $Q=\sum_i\lambda_i\ketbra{\psi_i}$ and $r\in\mathbb R$, we denote the spectral projections by
$
\{Q>r\}:=\sum_{\lambda_i>r}\ketbra{\psi_i}
$
and
$
\{Q\le r\}:=\sum_{\lambda_i\le r}\ketbra{\psi_i}.
$
The von Neumann entropy of a state $\rho \in {\mathcal D}(\mathcal H)$ is given by $S(\rho) := - \Tr \rho \log \rho$. Here and henceforth, all logarithms are taken to base $2$, and so $\exp(a)$ and $e^a$ are both understood to denote $2^a$.

A quantum channel
$
\pazocal E:\mathcal D(\mathcal H)\to\mathcal D(\mathcal K),
$
where $\mathcal H$ and $\mathcal K$ are finite-dimensional Hilbert spaces, is a linear completely positive trace-preserving (CPTP) map. We denote the identity channel by ${\rm id}:\mathcal D(\mathcal H)\to\mathcal D(\mathcal H)$.

For a state $\rho\in\mathcal D(\mathcal H)$ and a quantum channel $\pazocal C:\mathcal D(\mathcal H)\to\mathcal D(\mathcal H)$, the \emph{entanglement fidelity} of $\pazocal C$ with respect to $\rho$ is defined as
\begin{align}
F_e(\rho, \pazocal C)
:=
\bra{\psi_\rho}
({\rm id}\otimes\pazocal C)
\bigl(
\ket{\psi_\rho}\!\bra{\psi_\rho}
\bigr)
\ket{\psi_\rho},
\label{eq:ent-fid}
\end{align}
where $\ket{\psi_\rho}\in\mathcal H_R\otimes\mathcal H$ is any purification of $\rho$, i.e.\
$
\Tr_R
(
\ket{\psi_\rho}\!\bra{\psi_\rho}
)
=
\rho,
$
with $\mathcal H_R$ denoting the Hilbert space of a purifying reference system $R$. It can be shown that $F_e(\rho,\pazocal C)$ is independent of the choice of purification. If $\pazocal C$ has Kraus operators $\{E_k\}_k$, then (see e.g.~\cite{NC})
\begin{align}
F_e(\rho,\pazocal C)
=
\sum_k
|\Tr(\rho E_k)|^2.
\label{eq:NC1}
\end{align}}
\medskip

We now introduce the central notion of this paper.

\begin{Def}[(Weakly almost iid source)]
Let $\rho\in\mathcal D(\mathcal H)$. A sequence
$\bm{\rho}=(\rho_n)_n$, with
$\rho_n\in\mathcal D(\mathcal H^{\otimes n})$, is called a
\emph{weakly almost iid quantum source along $\rho$} if
\begin{equation}\label{eq:walmost iid}
\lim_{n\to\infty}
\mathbb E_{\substack{I\subseteq[n]\\|I|=k}}
\bigl\|
(\rho_n)_I-\rho^{\otimes k}
\bigr\|_1
=0
\quad \forall\, k\in\mathbb N_+,
\end{equation}
where the expectation is uniform over all subsets
$I\subseteq[n]$ of cardinality $k$.
\end{Def}

Condition~\eqref{eq:walmost iid} imposes only asymptotic agreement of fixed-size marginals, on average, with those of the tensor-power state $\rho^{\otimes n}$; it does not require the global state $\rho_n$ itself to be close to $\rho^{\otimes n}$ in trace norm. Consequently, weakly almost iid sources may still exhibit substantial long-range correlations and multipartite entanglement.

\subsection*{Some examples of weakly almost iid states}
\noindent
{{\bf Example 1: Haar-random pure states}
\smallskip

The following random construction shows that almost every Haar-random pure-state sequence is weakly almost i.i.d. In particular, deterministic examples of weakly almost iid sources with strong multipartite entanglement exist.

Let
\[
\rho_n=\ket{\psi_n}\!\bra{\psi_n},
\]
where the vectors
\(
\ket{\psi_n}\in\mathcal H^{\otimes n}
\)
are sampled independently according to the Haar measure, and
\(
\mathcal H\simeq\mathbb C^d
\),
with \(d\ge2\).
Fix \(k\in\mathbb N_+\) and consider \(n\ge k\). Let
\[
\mathcal H_I:=
\bigotimes_{i\in I}\mathcal H_i
\simeq
\mathcal H^{\otimes k},
\]
and let
\[
\psi_I
:=
\Tr_{I^c}\ket{\psi_n}\!\bra{\psi_n}
\in
\mathcal D(\mathcal H_I)
\]
be the reduced state associated with \(I\).

By unitary invariance of the Haar measure and permutation symmetry of the tensor factors, the distribution of the reduced state \(\psi_I\) depends only on \(|I|=k\). Fix such a subset
\(I\subseteq[n]\), with \(|I|=k\). Identifying
\(
\mathcal H^{\otimes n}
\simeq
\mathbb C^{d^k}
\otimes
\mathbb C^{d^{\,n-k}},
\)
the well-known formula for the expected purity of a Haar-random reduced state
(see, e.g.,~\cite{Lubkin1978}) yields
\begin{equation}\label{eq:page}
\mathbb E_{\psi_n}\Tr(\psi_I^2)
=
\frac{d^k+d^{\,n-k}}{d^n+1}.
\end{equation}

Hence, for fixed \(k\),
\bb
\mathbb E_{\psi_n}\Tr(\psi_I^2)
\longrightarrow
\frac1{d^k}
\qquad\text{as }n\to\infty,
\ee
which is precisely the purity of the maximally mixed state on
\(\mathcal H^{\otimes k}\).

For an operator \(X\) acting on a \(D\)-dimensional Hilbert space, using 
\(
\|X\|_1\le\sqrt D\,\|X\|_2,
\)
and the Jensen's inequality, together with
\bb
\Tr\!\left(
\psi_I-\frac{\id_k}{d^k}
\right)^2
=
\Tr(\psi_I^2)-\frac1{d^k},
\ee
yields
\bb\label{eq:Jen}
\mathbb E_{\psi_n}
\left\|
\psi_I-\frac{\id_k}{d^k}
\right\|_1
\le
\sqrt{
d^k
\left(
\mathbb E_{\psi_n}\Tr(\psi_I^2)
-
\frac1{d^k}
\right)
}.
\ee

Substituting~\eqref{eq:page} into~\eqref{eq:Jen} yields
\begin{equation}\label{eq:haar-est}
\mathbb E_{\psi_n}
\left\|
\psi_I-\frac{\id_k}{d^k}
\right\|_1
\le
\sqrt{
d^k
\left(
\frac{d^k+d^{\,n-k}}{d^n+1}
-
\frac1{d^k}
\right)
},
\end{equation}
which converges exponentially fast to zero as \(n\to\infty\) for fixed \(k\).

Therefore, for fixed \(k\),
\bb
\mathbb E_{\psi_n}
\mathbb E_{\substack{I\subseteq[n]\\|I|=k}}
\left\|
(\rho_n)_I-
\left(\frac{\id}{d}\right)^{\otimes k}
\right\|_1
\longrightarrow0
\qquad\text{as }n\to\infty.
\ee

Moreover, the exponential decay in~\eqref{eq:haar-est}, together with Markov's inequality and the first Borel--Cantelli lemma
(e.g.~\cite{Billingsley1995}), implies that, for every fixed \(k\), with probability one,
\bb
\mathbb E_{\substack{I\subseteq[n]\\|I|=k}}
\left\|
(\rho_n)_I-
\left(\frac{\id}{d}\right)^{\otimes k}
\right\|_1
\longrightarrow0
\qquad\text{as }n\to\infty.
\ee

Since \(\mathbb N_+\) is countable, the above convergence holds simultaneously for all fixed
\(k\in\mathbb N_+\) with probability one.

Hence, with probability one, the random sequence
\((\rho_n)_n\)
is weakly almost iid along
\(\id/d\).
In particular, there exist deterministic weakly almost iid sources consisting of globally pure, highly entangled multipartite states.
See Appendix A for details of the above.

\medskip}

\noindent
{{\bf{Example 2: Even parity states}}
For each $n$, define 
\bb
\rho_n
=
\frac{1}{2^{n-1}}
\sum_{\substack{x\in\{0,1\}^n\,;\, x=(x_1,\ldots, x_n)\\
x_1+\cdots+x_n\equiv0\;(\mathrm{mod}\;2)}}
|x\rangle\langle x|.
\ee
Then every proper marginal is exactly maximally mixed:
\bb
(\rho_n)_I
=
\left(\frac{I}{2}\right)^{\otimes |I|}
\qquad
\forall \, |I|<n.
\ee
Hence $(\rho_n)_n$ is weakly almost iid along $I/2$.
Nevertheless, the states possess nontrivial global correlations,
their support being contained in the even-parity subspace
\bb
\mathrm{Ran}\,(P_{\rm even}),
\quad \text{where}
\quad P_{\rm even}
=
\frac12(I+Z^{\otimes n}),
\ee
which is invisible to all fixed-size marginals.
This example illustrates that weakly almost iid states may exhibit
perfectly iid local behaviour while retaining strong long-range
constraints invisible to all fixed-size marginals. Thus, measurements confined to any proper subset of the sites cannot
distinguish $\rho_n$ from the maximally mixed product state
$(I/2)^{\otimes n}$. The global parity constraint can nevertheless
be detected by measuring all sites in the computational basis and
computing the parity of the outcomes.

The even parity state is fully separable, being a uniform mixture of computational basis product states. Nevertheless, it is globally distinguishable from the maximally mixed product state due to the presence of a hidden parity constraint. This example illustrates that weakly almost i.i.d. states may exhibit highly nontrivial long-range correlations even within the fully separable sector. 
\medskip

\noindent
{\bf Example 3: Fixed-composition states}
\smallskip

The following gives a simple deterministic example of a weakly almost iid source
along a non-maximally mixed reference state. Let
\bb
\rho
=
p\ket{0}\!\bra{0}
+
(1-p)\ket{1}\!\bra{1},
\qquad
0<p<1, \,\, p \ne 1/2,
\ee
and let $(m_n)_n$ be a sequence of integers, \(m_n\in\{0,\ldots,n\}\), such that
\bb
\frac{m_n}{n}\longrightarrow p
\quad\text{as } \quad n\to\infty.
\ee
Define $\rho_n$ to be the uniform mixture over all computational-basis strings
containing exactly $m_n$ zeros and $n-m_n$ ones, i.e.
\bb
\rho_n
=
\frac{1}{\binom{n}{m_n}}
\sum_{\substack{x^n\in\{0,1\}^n\\ N_0(x^n)=m_n}}
\ket{x^n}\!\bra{x^n},
\ee
where $N_0(x^n)$ denotes the number of zeros in $x^n$.

The state $\rho_n$ is permutation invariant, and hence its $k$-site marginal is
independent of the choice of $I\subseteq[n]$ with $|I|=k$. For fixed $k$, this
marginal corresponds to sampling $k$ symbols without replacement from a population
containing $m_n$ zeros and $n-m_n$ ones. Since $m_n/n\to p$, the corresponding
hypergeometric distribution converges, for every fixed $k$, to the product
Bernoulli distribution with parameter $p$. Consequently,
\bb
\left\|
(\rho_n)_I-\rho^{\otimes k}
\right\|_1
\longrightarrow0
\qquad\text{as }n\to\infty
\ee
for every fixed $k$, and therefore
\bb
\mathbb E_{\substack{I\subseteq[n]\\ |I|=k}}
\left\|
(\rho_n)_I-\rho^{\otimes k}
\right\|_1
\longrightarrow0.
\ee
Thus $(\rho_n)_n$ is weakly almost iid along $\rho$.
\medskip

The above examples demonstrate that weakly almost iid states can exhibit radically different global structures, ranging from highly entangled random pure states to fully separable yet strongly classically correlated states.
}

\section{The key lemmas}\label{sec:key-lemmas}

We now prove the two basic concentration lemmas on which the subsequent applications rest.

\begin{lemma}[(Noncommutative weakly almost iid law of large numbers)]
\label{lem:nc-weak-almost iid-lln}
Let $\rho\in\mathcal D(\mathcal H)$, where $\mathcal H \simeq \mathbb{C}^d$ with $d < \infty$. Let
$\bm{\rho}=(\rho_n)_{n\ge1}$, with
$\rho_n\in\mathcal D(\mathcal H^{\otimes n})$, be weakly almost iid along $\rho$.
Let $A\in\mathcal L(\mathcal H)$ be self-adjoint, define
$\mu:=\Tr(\rho A)$ and
$\overline A_n:=\frac{1}{n}\sum_{i=1}^n A^{(i)}$, where
$A^{(i)}:=\id^{\otimes(i-1)}\otimes A\otimes \id^{\otimes(n-i)}$.
Then $\Tr(\rho_n\overline A_n)\to\mu$ and
$\Tr[\rho_n(\overline A_n-\mu \id_n)^2]\to0$ as $n\to\infty$.
Consequently, for every $\delta>0$,
\begin{equation}\label{eq:nc-cheb-conc}
\Tr\!\left[
\rho_n
\left\{
|\overline A_n-\mu \id_n|>\delta
\right\}
\right]
\to0
\quad \text{ as } n\to\infty.
\end{equation}
\end{lemma}

\begin{remark}
{The lemma shows that empirical averages of one-site observables exhibit asymptotic spectral concentration around their expectation values with respect to the reference state $\rho$. It is the noncommutative analogue of convergence in probability in the classical weak law of large numbers.}
\end{remark}

{
\begin{proof}
Set $M:=\|A\|_\infty$. Since $A^{(i)}$ acts only on the $i$-th tensor factor,
\bb
\Tr(\rho_n\overline A_n)-\mu
=
\frac1n\sum_{i=1}^n
\Tr\!\bigl[((\rho_n)_i-\rho)A\bigr],
\ee
and hence
\bb\label{eq:convg1}
\bigl|
\Tr(\rho_n\overline A_n)-\mu
\bigr|
\le
M\,\frac1n\sum_{i=1}^n\|(\rho_n)_i-\rho\|_1
\longrightarrow0 \quad \text{as} \quad n \to \infty,
\ee
by the weakly almost iid assumption with $k=1$.

For the second moment, write
$\overline A_n^2=n^{-2}\sum_i(A^{(i)})^2+2n^{-2}\sum_{i<j}A^{(i)}A^{(j)}$.
{The contribution of the first sum to \(\Tr(\rho_n\overline A_n^2)\) is bounded in absolute value by
\bb
\frac1{n^2}\sum_{i=1}^n
\bigl|\Tr[\rho_n(A^{(i)})^2]\bigr|
\le
\frac1{n^2}\sum_{i=1}^n \|A\|_\infty^2
=
\frac{M^2}{n},
\ee
and therefore vanishes as \(n\to\infty\).}

For the off-diagonal part,
$\Tr(\rho_nA^{(i)}A^{(j)})=\Tr((\rho_n)_{ij}(A\otimes A))$, whence
\bb
\frac2{n^2}\sum_{i<j}\Tr(\rho_nA^{(i)}A^{(j)})
=
\frac{n(n-1)}{n^2}\mu^2+E_n,
\ee
{where the error term \(E_n\) satisfies}
\bb
|E_n|
\le
\frac{n(n-1)}{n^2}M^2
\left[
\frac1{\binom n2}
\sum_{i<j}
\|(\rho_n)_{ij}-\rho^{\otimes2}\|_1
\right].
\ee
The bracketed average tends to zero by the assumption with $k=2$, so
$\Tr(\rho_n\overline A_n^2)\to\mu^2$ as $n \to \infty$. This, together with the first-moment convergence
gives $\Tr[\rho_n(\overline A_n-\mu \id_n)^2]\to0$ as $n \to \infty$.

Finally, by functional calculus,
$(\overline A_n-\mu \id_n)^2\ge
\delta^2\{|\overline A_n-\mu \id_n|>\delta\}$.
Taking the trace against $\rho_n$ gives the noncommutative Chebyshev bound
\bb
\Tr\!\left[
\rho_n\{|\overline A_n-\mu \id_n|>\delta\}
\right]
\le
\delta^{-2}\Tr[\rho_n(\overline A_n-\mu \id_n)^2],
\ee
and the right-hand side tends to zero in the limit $n \to \infty$.
\end{proof}
}

\medskip

\begin{lemma}[(Universal typical projector)]
\label{lem:universal-typical-projector}
Let $\rho\in\mathcal D(\mathcal H)$, with $\dim\mathcal H=d<\infty$, and let
$\bm{\rho}=(\rho_n)_n$, with $\rho_n\in\mathcal D(\mathcal H^{\otimes n})$, be weakly almost iid along $\rho$.
Fix $q\in(0,1)$, set
\bb
\rho_q:=(1-q)\rho+q\id/d,
\ee
$A_q:=-\log\rho_q$, and
$h_q:=\Tr(\rho A_q)$.
For $\delta>0$, define the spectral projector 
\begin{equation}\label{eq:Pi-n}
\Pi_n
:=
\left\{
-\frac1n\log\rho_q^{\otimes n}
\le
h_q+\delta
\right\}.
\end{equation}
Then $\Tr(\Pi_n\rho_n)\to1$ as $n\to\infty$, while
$\Tr\Pi_n\le e^{n(h_q+\delta)}$ for every $n \in {\mathbb{N}}$. Moreover,
$h_q\to S(\rho)$ as $q\to0$.
\end{lemma}

\begin{remark}
{This lemma is the basic entropy-concentration result underlying the paper.
It shows that, for every fixed $\varepsilon>0$, every weakly almost iid
source along $\rho$ asymptotically concentrates on a common sequence
of subspaces whose dimensions are at most $2^{n(S(\rho)+\varepsilon)}$.
These subspaces depend only on $\rho$ and $\varepsilon$, and not on
the individual source.}
\end{remark}

{
\begin{proof}
Since $\rho_q$ is full rank, $A_q=-\log\rho_q$ is bounded. Then
\bb
\overline A_{q,n}:=\frac{1}{n}\sum_{i=1}^n A_q^{(i)}=-\frac{1}{n}\log \rho_q^{\otimes n},
\ee
and therefore
$\Pi_n=\{\overline A_{q,n}\le h_q+\delta\}$.

Applying Lemma~\ref{lem:nc-weak-almost iid-lln} to $A_q$ gives for every $\delta >0$,
\bb
\Tr\!\left[
\rho_n
\{|\overline A_{q,n}-h_q \id_n|>\delta\}
\right]\to0 \quad \text{as} \quad n \to \infty.
\ee
Since
$\{\overline A_{q,n}>h_q+\delta\}\subseteq
\{|\overline A_{q,n}-h_q \id_n|>\delta\}$,
we obtain 
\bb
\Tr[\rho_n(\id_n-\Pi_n)]\to0 \quad \text{ as } \quad n \to \infty,
\ee
and hence $\Tr(\Pi_n\rho_n)\to1$ in this limit.

{For the rank bound, observe that, by definition of \(\Pi_n\), the spectrum of
\(\overline A_{q,n}\) on \(\operatorname{Ran}\Pi_n\) is bounded above by \(h_q+\delta\). Hence, by functional calculus,
\bb
\Pi_n e^{-n\overline A_{q,n}}\Pi_n
\ge
e^{-n(h_q+\delta)}\Pi_n .
\ee
Since \(e^{-n\overline A_{q,n}}=\rho_q^{\otimes n}\), it follows that
\bb
1
\ge
\operatorname{Tr}(\rho_q^{\otimes n}\Pi_n)
=
\operatorname{Tr}(\Pi_n e^{-n\overline A_{q,n}}\Pi_n)
\ge
e^{-n(h_q+\delta)}\operatorname{Tr}\Pi_n .
\ee
Therefore,
\bb
\Tr\Pi_n\le e^{n(h_q+\delta)} \quad \text{for all} \quad n \in {\mathbb{N}}.
\ee}

Finally, $\rho_q$ commutes with $\rho$. If 
$\rho=\sum_i\lambda_i\ketbra{i}$ denotes the eigenvalue decomposition of $\rho$, then
\bb
h_q
=
-\sum_{\lambda_i>0}
\lambda_i
\log\!\left((1-q)\lambda_i+\frac qd\right),
\ee
and letting $q\to 0$ gives
$h_q\to-\sum_{\lambda_i>0}\lambda_i\log\lambda_i=S(\rho)$.
\end{proof}}

\section{Applications of the Concentration Lemmas}\label{sec:applns}
{
\subsection{Universal quantum data compression for weakly almost iid sources}\label{sec:data-comp}

{The universal entropy-concentration principle of Lemma~\ref{lem:universal-typical-projector} yields the universal quantum data compression result stated in Theorem~\ref{thm:univ-data-comp} below. In contrast to the approach of~\cite{girardi2026quantumshannontheoryrobust}, where this result was derived indirectly through a reduction to hypothesis testing within a broader robustness framework, the present proof shows that compression follows directly from entropy concentration via the universal typical-projector construction.}

\subsubsection{Individual versus universal optimal rates}

{We first distinguish between (1) compression schemes tailored to a single
source and (2) universal schemes that reliably compress every weakly almost
iid source along the same reference state using the same sequence of
encoding and decoding channels.}

Let $\rho\in\mathcal D(\mathcal H)$, with $\dim\mathcal H<\infty$, and let $\mathfrak S_\rho$ denote the class of all weakly almost iid quantum sources along $\rho$.

{\begin{Def}{(Optimal rate for an individual weakly almost iid source along $\rho$)}

Let $\bm{\rho}=(\rho_n)_{n} \in \mathfrak S_\rho$, with $\rho_n\in\mathcal D(\mathcal H^{\otimes n})$.
A number \(R\ge0\) is said to be achievable for the source
\(\bm\rho\) if there exist Hilbert spaces \(\mathcal K_n\subseteq \mathcal H^{\otimes n}\) and
encoding (compression) and decoding (decompression) channels
\bb
\pazocal E_n:
\mathcal D(\mathcal H^{\otimes n})
\to
\mathcal D(\mathcal K_n)
\quad \text{and} \quad 
\pazocal D_n:
\mathcal D(\mathcal K_n)
\to
\mathcal D(\mathcal H^{\otimes n}),
\ee
such that
\bb
\limsup_{n\to\infty}
\frac1n\log\dim\mathcal K_n
\le R,
\ee
and
\bb
F_e\!\left(
\rho_n,
\pazocal D_n\circ\pazocal E_n
\right)
\longrightarrow1 \quad \text{as} \quad n \to \infty.
\ee
The optimal compression rate of the source
\(\bm\rho\) is defined by
\bb
R_{\mathrm{opt}}(\bm\rho)
:=
\inf
\left\{
R:
R \text{ is achievable for }
\bm\rho
\right\}.
\ee
\end{Def}}

\begin{Def}{(Optimal universal compression rate for $\mathfrak S_\rho$)}

A number \(R\ge0\) is said to be a {\em{universal}} achievable rate of data compression for $\mathfrak S_\rho$ if there exist sequences of finite-dimensional Hilbert spaces $\mathcal K_n \subseteq {\mathcal H}^{\otimes n}$, and encoding and decoding (CPTP) maps
$$\pazocal E_n:\mathcal D(\mathcal H^{\otimes n})\to\mathcal D(\mathcal K_n) \quad \text{and} \quad
\pazocal D_n:\mathcal D(\mathcal K_n)\to\mathcal D(\mathcal H^{\otimes n}),$$
{\em{depending only on $R$ and the reference state $\rho$}}, but not on the individual source,
such that {\em{for every}} ${\bm{\omega}} =(\omega_n)_n \in \mathfrak S_\rho$, 
\bb
\limsup_{n\to\infty}\frac{1}{n}\log\dim\mathcal K_n\le R \quad \text{and} \quad F_e(\omega_n,\pazocal D_n\circ\pazocal E_n)\to1 \quad \text{as} \quad n \to \infty.
\ee
Any such compression scheme $({\mathcal E_n}, {\mathcal D_n})_n$ is said to be a {\em{universally reliable}} compression scheme for $\mathfrak S_\rho.$ The optimal {\em{universal}} compression rate for $\mathfrak S_\rho$ is given by
\bb
R_{\rm opt}^{\rm univ}(\mathfrak S_\rho):= \inf\{R \, : \, R \text{\, is a universal achievable rate for\,} {\mathfrak S}_\rho\}
\ee
\end{Def}

By definition,
$R_{\rm opt}^{\rm univ}(\mathfrak S_\rho)\ge
\sup_{\bm\omega\in\mathfrak S_\rho}R_{\rm opt}(\bm\omega)$,
and the inequality may be strict.

\begin{thm}\label{thm:univ-data-comp}
Let $\rho\in\mathcal D(\mathcal H)$, with $\dim\mathcal H<\infty$. Then
\bb\label{eq:univ-rate}
R_{\rm opt}^{\rm univ}(\mathfrak S_\rho)=S(\rho).
\ee
\end{thm}

\begin{proof}
Fix $R>S(\rho)$. By Lemma~\ref{lem:universal-typical-projector}, there exist $q\in(0,1)$, $\delta>0$, and spectral projectors
\bb
\Pi_n=\{-\frac{1}{n}\log\rho_q^{\otimes n}\le h_q+\delta\},
\ee
where $\rho_q=(1-q)\rho+q\id/d$, such that
for every $\bm\omega\in\mathfrak S_\rho$, 
\bb
\Tr(\Pi_n\omega_n)\to1 \quad \text{as} \quad n \to \infty,
\ee 
while
$\Tr\Pi_n\le e^{n(h_q+\delta)}$ and
$h_q\to S(\rho)$ as $q\to0$.
Let us choose $q$ and $\delta$ so that $h_q+\delta<R$.

Define
$\mathcal K_n:={\rm Ran}\,\Pi_n$.
Let $\tau_n$ be any fixed state in ${\mathcal D}(\mathcal K_n)$, and define for any $\nu \in {\mathcal D}(\mathcal H^{\otimes n})$,
\bb
\pazocal E_n(\nu):=\Pi_n\nu \Pi_n+\Tr[(\id_-\Pi_n)\nu]\tau_n.
\ee
Let \(\pazocal D_n\) denote the decoding map induced by the inclusion
\(\mathcal K_n\subseteq\mathcal H^{\otimes n}\), i.e.\ extending states on
\(\mathcal K_n\) by zero on \(\mathcal K_n^\perp\).
Then
\bb\frac{1}{n}\log\dim\mathcal K_n\leq h_q + \delta <R \quad {\text{for all}}\, \,n \in {\mathbb{N}}.\ee
{Let
$
\pazocal C_n:=\pazocal D_n\circ\pazocal E_n
$
denote the overall compression--decompression channel. The channel \(\pazocal C_n\) has a Kraus representation with \(\Pi_n\) as one of its Kraus operators. Hence, \eqref{eq:NC1} implies that, for every
\(\bm\omega\in\mathfrak S_\rho\), and for every $n \in {\mathbb{N}}$,
\bb
F_e(\omega_n,\pazocal C_n)
\ge
|\Tr(\omega_n\Pi_n)|^2.
\ee
Since \(\Tr(\omega_n\Pi_n)\to1\) as $n \to \infty$, it follows that
$
F_e(\omega_n,\pazocal C_n)\to1$ in this limit.}
Hence, every rate $R>S(\rho)$ is universally achievable, i.e.
\bb\label{eq:one-way}
R_{\rm opt}^{\rm univ}(\mathfrak S_\rho) \leq S(\rho)
\ee
\smallskip

{Conversely, note that the iid source \((\rho^{\otimes n})_n\) belongs to the class
\(\mathfrak S_\rho\). Hence any universally reliable compression scheme for
\(\mathfrak S_\rho\) must, in particular, compress the iid source
\((\rho^{\otimes n})_n\) reliably. By the strong converse part of Schumacher's theorem
(see e.g.~\cite{Winter1999} and references therein), any such sequence of compression
schemes, with compressed Hilbert spaces \((\mathcal K_n)_n\), must satisfy
\bb
\liminf_{n\to\infty}
\frac1n\log\dim\mathcal K_n
\ge S(\rho).
\ee
Therefore \(R_{\rm opt}^{\rm univ}(\mathfrak S_\rho)\ge S(\rho)\). Together with
\eqref{eq:one-way}, this proves the theorem.}
\end{proof}

The preceding theorem concerns universal compression over the entire class $\mathfrak S_\rho$. By contrast, the optimal compression rate of an individual weakly almost iid source ${\bm{\rho}} \in {\mathfrak S_\rho}$ may be strictly smaller than the von Neumann entropy $S(\rho)$ of the reference state.

Indeed, let $\rho=\id/d$, so that $S(\rho)=\log d$. By Schumacher’s theorem~\cite{Schumacher}, the genuinely iid source $(\rho^{\otimes n})_n$ has optimal compression rate $\log d$. On the other hand, as discussed in Section~\ref{sec:walmost iid}, there exist globally pure states
$\rho_n=\ket{\psi_n}\!\bra{\psi_n}$
such that $(\rho_n)_n$ is weakly almost iid along $\id/d$. Since the encoder and decoder are allowed to depend on the individual source, one may compress $\rho_n$ into a one-dimensional Hilbert space. Indeed, let $\ket0$ span a one-dimensional compressed space, define
$\pazocal E_n(\omega):=\Tr(\omega)\ket0\!\bra0$
and
$\pazocal D_n(\ket0\!\bra0):=\rho_n$\footnote
{
More precisely, letting \(\mathcal K_n:=\operatorname{span}\{\ket0\}\), define
$
\pazocal E_n(\omega):=\Tr(\omega)\ket0\!\bra0$, and 
$\pazocal D_n(\tau):=\Tr(\tau)\rho_n$ for any $\tau\in\mathcal D(\mathcal K_n).
$}.
Then
$(\pazocal D_n\circ\pazocal E_n)(\rho_n)=\rho_n$,
and hence for all $n \in {\mathbb{N}},$
$F_e(\rho_n,\pazocal D_n\circ\pazocal E_n)=1$.
Therefore
$R_{\rm opt}(\bm\rho)=0$,
even though the reference entropy is
$S(\rho)=\log d$.
}
\subsection{Asymmetric binary hypothesis testing}\label{sec:HT}

We now apply Lemmas~\ref{lem:nc-weak-almost iid-lln} and~\ref{lem:universal-typical-projector} to asymmetric binary quantum hypothesis testing. The resulting theorem (Theorem~\ref{thm:walmost iid-stein}) may be viewed as a robustness version of the direct part of quantum Stein’s lemma~\cite{HiaiPetz1991, OgawaNagaoka2000}: although the null hypothesis is only weakly almost i.i.d., the optimal universal type-II error exponent (i.e.\ the quantum Stein exponent) remains unchanged.

Theorem~\ref{thm:walmost iid-stein} also illustrates the strength of the entropy-concentration framework developed in Section~\ref{sec:key-lemmas}. The same robustness result for quantum hypothesis testing was previously obtained in~\cite{girardi2026quantumshannontheoryrobust}, but through a substantially more involved analysis embedded within a broader operational robustness framework. In contrast, the present proof is direct and conceptually transparent: the desired tests arise immediately from the concentration properties established in Lemmas~\ref{lem:nc-weak-almost iid-lln} and~\ref{lem:universal-typical-projector}. In particular, the resulting tests are universal, depending only on the reference states $\rho$ and $\sigma$, and not on the individual weakly almost iid source. 
\medskip

As in the previous section, let $\rho\in\mathcal D(\mathcal H)$, with $\dim\mathcal H<\infty$, and let $\mathfrak S_\rho$ denote the class of all weakly almost iid quantum sources ${\bm{\rho}}= (\rho_n)_n$ along $\rho$.
We consider the following binary hypothesis testing problem:
\begin{itemize}
    \item {\bf{Null hypothesis}} $H_0:$ the source belongs to the class $\mathfrak S_\rho$.
    \item {\bf{Alternative hypothesis}} $H_1:$ the source is an i.i.d.~source $\sigma^{\otimes n}$,
where $ \sigma \in {\mathcal D}({\mathcal H}) $.

\end{itemize}
If $0 \le T_n \leq {\id}_n$ defines the binary POVM used to distinguish between these two hypotheses, then for any ${\bm{\rho}}= (\rho_n)_n \in {\mathfrak S_\rho}$, the probabilities of type-I and type-II errors are, respectively, given by
\bb
\alpha_n := \Tr[({\id}_n -T_n) \rho_n]\quad ; \quad 
\beta_n := \Tr[T_n \sigma^{\otimes n}]
\ee
In the asymmetric setting, one minimizes the probability of type-II error ($\beta_n$) under the constraint that the probability of type-I error ($\alpha_n$) is below some fixed threshold $\varepsilon \in (0,1)$. 
The quantity of interest is the optimal asymptotic type-II error exponent (also called the quantum Stein exponent), which can be expressed in terms of the hypothesis testing relative entropy as follows:
\bb\label{eq:stexpo} 
\liminf_{n\to\infty}\frac1n D_H^\varepsilon(\rho_n\|\sigma^{\otimes n}),\ee
where for any two states $\tau, \omega$ on a finite-dimensional Hilbert space,
{\bb D_H^\varepsilon(\tau\|\omega)
:=
-\log
\inf\left\{
\Tr(T\omega):
0\le T\le \id,\ 
\Tr[(\id-T)\tau]\le\varepsilon
\right\}.\ee}
The celebrated quantum Stein's lemma~\cite{OgawaNagaoka2000, Hayashi_2007} states that if the null hypothesis is also iid (say, given by $\rho^{\otimes n}$) and ${\rm supp} \, \rho \subseteq {\rm supp} \, \sigma$, then the quantum Stein exponent given by \eqref{eq:stexpo} is equal to the Umegaki relative entropy~\cite{Umegaki1962}:
\bb
D(\rho\| \sigma) := \Tr[ \rho(\log \rho - \log \sigma)]
\ee

\begin{thm}
\label{thm:walmost iid-stein}
Let $\rho,\sigma\in\mathcal D(\mathcal H)$, with $\dim\mathcal H<\infty$ and $\sigma>0$, and let $\bm\rho = (\rho_n)_n\in \mathfrak S_\rho$. Then, for every $\varepsilon\in(0,1)$,
\bb\label{eq:stein-walmost iid}
\liminf_{n\to\infty}\frac1n D_H^\varepsilon(\rho_n\|\sigma^{\otimes n})
\ge D(\rho\|\sigma).
\ee
More precisely, for every \(R<D(\rho\|\sigma)\), there exists a sequence of tests
\((T_n)_n\), depending only on \(\rho\), \(\sigma\) and $R$, such that, for every
\(\bm\rho=(\rho_n)_n\in\mathfrak S_\rho\),
\bb
\Tr(T_n\rho_n)\to1 \quad \text{as} \quad n \to \infty
\quad\text{while}\quad
\Tr(T_n\sigma^{\otimes n})\le e^{-nR}
\,\,\text{for all} \,\, n \in {\mathbb{N}}.\ee
\end{thm}

{\begin{proof}
Set
\(
A:=-\log\sigma
\)
and
\(
a:=\Tr(\rho A)=\Tr[\rho(-\log\sigma)].
\)
For \(\delta>0\), define the spectral projector
\bb
Q_n
:=
\left\{
-\frac1n\log\sigma^{\otimes n}
\ge
a-\delta
\right\}
=
\left\{
\overline A_n
\ge
a-\delta
\right\},
\ee
where
\(
\overline A_n
=
\frac1n\sum_{i=1}^n A^{(i)}.
\) Thus \(Q_n\) is precisely the spectral projector associated with the
empirical observable appearing in
Lemma~\ref{lem:nc-weak-almost iid-lln}. Hence
\(
\Tr(Q_n\rho_n)\to1
\)
as \(n\to\infty\).
Intuitively, \(Q_n\) projects onto the spectral region in which
\(\sigma^{\otimes n}\) is exponentially suppressed at rate
approximately \(a\).

Next, fix \(q\in(0,1)\). By
Lemma~\ref{lem:universal-typical-projector}, there exists a sequence of
projections \(P_n\) such that
\(
\Tr(P_n\rho_n)\to1
\)
as \(n\to\infty\), while
\bb
\Tr P_n
\le
e^{n(h_q+\delta)}
\qquad
\forall n\in\mathbb N,
\ee
where
\(
h_q
:=
\Tr[\rho(-\log\rho_q)]
\)
and
\(
\rho_q
:=
(1-q)\rho+q\frac{\id}{d}.
\)
The projections \(P_n\) provide entropy concentration: they asymptotically capture all of the mass of \(\rho_n\), while the dimensions of their ranges grow at most exponentially at rate \(h_q\).

We now combine the spectral concentration information encoded in \(Q_n\) with the entropy-concentration information encoded in \(P_n\). Define
\bb
T_n:=Q_nP_nQ_n.
\ee
Since \(P_n\) and \(Q_n\) are projections,
\(
0\le T_n\le \id_n,
\)
and hence
\(
\{T_n,\id_n-T_n\}
\)
defines a valid binary POVM. Since
\(
\Tr(Q_n\rho_n)\to1
\)
as \(n\to\infty\), the gentle operator lemma
(see e.g.~Lemma~9.4.2 of~\cite{Wilde2017}) implies
\bb
\|Q_n\rho_nQ_n-\rho_n\|_1\to0 \quad \text{as} \quad n \to \infty.
\ee
Moreover,
$
\Tr(T_n\rho_n)
=
\Tr(P_nQ_n\rho_nQ_n)
=
\Tr(P_n\rho_n)
+
\Tr\!\left[
P_n(Q_n\rho_nQ_n-\rho_n)
\right].
$
Since \(\|P_n\|_\infty\le1\),
\bb
\left|
\Tr\!\left[
P_n(Q_n\rho_nQ_n-\rho_n)
\right]
\right|
\le
\|Q_n\rho_nQ_n-\rho_n\|_1.
\ee
Therefore,
\bb
\Tr(T_n\rho_n)
\ge
\Tr(P_n\rho_n)
-
\|Q_n\rho_nQ_n-\rho_n\|_1.
\ee
Since
\(
\Tr(P_n\rho_n)\to1
\)
and
\(
\|Q_n\rho_nQ_n-\rho_n\|_1\to0,
\)
it follows that
\(
\Tr(T_n\rho_n)\to1
\)
as \(n\to\infty\). Hence, the probability of type-I error for the sequence of binary POVMs defined by $(T_n)_n$ approaches $0$, asymptotically.

Next, we analyze the asymptotic behaviour of the corresponding probability of type-II error. Note that \(Q_n\) commutes with \(\sigma^{\otimes n}\). Moreover, on
\(\Ran Q_n\),
\(
-\frac1n\log\sigma^{\otimes n}\ge a-\delta.
\)
Hence, by functional calculus,
\bb
Q_n\sigma^{\otimes n}Q_n
\le
e^{-n(a-\delta)}Q_n.
\ee
Therefore,
\bb
\Tr(T_n\sigma^{\otimes n})
=
\Tr(Q_nP_nQ_n\sigma^{\otimes n})
\le
e^{-n(a-\delta)}\Tr(P_n)
\le
e^{-n(a-h_q-2\delta)}.
\ee
Finally,
\(
a-h_q
\to
\Tr[\rho(-\log\sigma)]-S(\rho)
=
D(\rho\Vert\sigma)
\)
as \(q\to0\). Hence, choosing \(q\) and \(\delta\) sufficiently small
yields
\bb
\Tr(T_n\sigma^{\otimes n})
\le
e^{-nR}
\qquad
\forall n\in\mathbb N
\ee
for any \(R<D(\rho\Vert\sigma)\).

Since
\(
\Tr(T_n\rho_n)\to1,
\)
for every \(\varepsilon\in(0,1)\), the POVM
\(
\{T_n,\id_n-T_n\}
\)
is feasible in the definition of
\(D_H^\varepsilon(\rho_n\Vert\sigma^{\otimes n})\) for all sufficiently
large \(n\). Hence
\(
D_H^\varepsilon(\rho_n\Vert\sigma^{\otimes n})
\ge
nR
\)
for all sufficiently large \(n\). Dividing by \(n\), taking the limit
inferior, and then letting
\(
R\uparrow D(\rho\Vert\sigma)
\)
proves~\eqref{eq:stein-walmost iid}.
\end{proof}}

{\begin{remark}
The theorem is optimal in the universal setting. Indeed, the iid source
\((\rho^{\otimes n})_n\) belongs to \(\mathfrak S_\rho\), so any universal testing strategy
must in particular work for the genuinely iid null hypothesis. The strong converse part of
quantum Stein's lemma therefore rules out any universal exponent larger than
\(D(\rho\|\sigma)\). Thus, in the universal setting, the quantum Stein exponent is the same as in
the iid case.
\end{remark}}

\begin{remark}
By contrast, an individual-source converse fails in general. Take
$\rho=\sigma=\id/d$, so that $D(\rho\|\sigma)=0$, and let
$\rho_n=\ket{\psi_n}\!\bra{\psi_n}$, where
$\ket{\psi_n}\in\mathcal H^{\otimes n}$ is Haar-random and
$d:=\dim\mathcal H$. As shown in Section~\ref{sec:walmost iid}, with probability one the sequence $(\rho_n)_n$ is weakly almost iid along $\id/d$. Choosing the source-dependent test $T_n:=\rho_n$, we have
$\Tr(T_n\rho_n)=1$ and $\Tr(T_n\sigma^{\otimes n})=\frac{1}{d^{n}}$, and hence for any $\varepsilon \in (0,1),$
\bb
\liminf_{n\to\infty}\frac1n
D_H^\varepsilon(\rho_n\|\sigma^{\otimes n})
\ge
\log d
>
D(\rho\|\sigma).
\ee
Thus no individual-source converse of the form
$\limsup_{n\to\infty}\frac{1}{n}D_H^\varepsilon(\rho_n\|\sigma^{\otimes n})\le D(\rho\|\sigma)$
can hold for arbitrary weakly almost iid null hypotheses.
\end{remark}
\begin{remark}
Note that our proof of Theorem~\ref{thm:walmost iid-stein} is tailored to the setting in which the weakly almost iid source appears in the null hypothesis. Indeed, Lemma~\ref{lem:universal-typical-projector} provides universal projectors onto subspaces capturing asymptotically all of the weight of every weakly almost iid source along a fixed reference state, which is precisely what is needed to ensure asymptotically vanishing type-I error.

If instead the null hypothesis is iid and the alternative hypothesis is weakly almost i.i.d., one would need upper bounds on quantities of the form \(\Tr(T_n\rho_n)\), where \((\rho_n)_n\) is weakly almost iid and the tests \((T_n)_n\) are designed to accept the iid null hypothesis $\sigma^{\otimes n}$. Our concentration lemmas do not yield such bounds, since the weakly almost iid structure constrains only fixed-size marginals while permitting arbitrary long-range correlations in the global alternative states.

\end{remark}

{
\subsection{Many-body concentration and generalized Gibbs ensembles}\label{sec:many-body}

{We next apply Lemma~\ref{lem:nc-weak-almost iid-lln} to quantum many-body systems. Here
\(\mathcal H^{\otimes n}\) is interpreted as the Hilbert space of an \(n\)-body system with local
Hilbert space \(\mathcal H\), and a weakly almost iid source
\(\bm\rho=(\rho_n)_n\) describes a sequence of correlated many-body states whose fixed-size
local marginals asymptotically agree, on average, with those of the product state
\(\rho^{\otimes n}\). The following theorem shows that, despite possible long-range
correlations and multipartite entanglement, empirical averages of finitely many commuting
one-site observables concentrate around the values predicted by the reference state  $\rho \in {\mathcal D}({\mathcal H})$.}

\begin{thm}[(Thermodynamic typicality)]
\label{thm:thermodynamic-typicality}
Let $\bm\rho=(\rho_n)_n$, with $\rho_n\in\mathcal D(\mathcal H^{\otimes n})$, be weakly almost iid along $\rho\in\mathcal D(\mathcal H)$. Let
$A_1,\ldots,A_m\in\mathcal L(\mathcal H)$ be mutually commuting self-adjoint observables, and define
$a_j:=\Tr(\rho A_j)$ and 
$\overline A_{j,n}:=\frac1n\sum_{i=1}^n A_j^{(i)},
$ where
\(
A_j^{(i)}
=
\id^{\otimes(i-1)}\otimes A_j\otimes \id^{\otimes(n-i)}
\).
For $\delta>0$, let
$P_{j,n}:=\{|\overline A_{j,n}-a_j\id_n|\le\delta\}$.
Then
\bb
\Tr\!\left[
\rho_n\prod_{j=1}^m P_{j,n}
\right]\to1
\quad \text{as} \quad n\to\infty.
\ee
Equivalently, for every fixed $\delta>0$, the probability that at least one empirical observable $\overline A_{j,n}$ takes a value outside the interval $[a_j-\delta,a_j+\delta]$ tends to zero as $n\to\infty$.
\end{thm}

The theorem says that, for any finite family of commuting one-site observables, the corresponding empirical observables concentrate around the values predicted by the reference state $\rho$, even for highly correlated many-body states which are weakly almost iid along $\rho$. {Hence, for any fixed finite family of commuting single-site observables,
weakly almost iid many-body states exhibit concentration around the
same expectation values as the corresponding iid state
$\rho^{\otimes n}$, despite possible long-range correlations and
multipartite entanglement.}
\begin{remark}
In a spin-$\frac12$ system with local Hilbert space
$\mathcal H=\mathbb C^2$, one may take
\bb
A_1=\sigma_z,\,\,
A_2=\frac{\id+\sigma_z}{2}
=
|\uparrow\rangle\!\langle\uparrow|.
\ee
The corresponding empirical observables describe respectively the average magnetization per site in the $z$-direction and the density of spin-up sites.

For a spin-$1$ lattice, with local Hilbert space $\mathcal H=\mathbb C^3$, one may instead take
\bb
A_1=S_z, \,\,A_2=S_z^2,
\ee
where $S_z$ denotes the usual spin-$z$ operator. The associated empirical observables correspond respectively to the magnetization density and quadrupole density. Theorem~\ref{thm:thermodynamic-typicality} then implies that these empirical observables concentrate around the values predicted by the reference state $\rho$, even though the underlying many-body states need not be close to the iid state $\rho^{\otimes n}$.
\end{remark}

\begin{proof}
Since the observables $A_1,\ldots,A_m$ commute, so do the empirical observables $\overline A_{1,n},\ldots,$ $\overline A_{m,n}$ and their spectral projections $P_{1,n},\ldots,P_{m,n}$. From Lemma~\ref{lem:nc-weak-almost iid-lln} we have that, for each $j \in \{1, \ldots, m\},$
$\Tr[\rho_n(\id_n-P_{j,n})]\to0$
as $n \to \infty$. Using the operator inequality
$\id_n-\prod_{j=1}^m P_{j,n}\le\sum_{j=1}^m(\id_n-P_{j,n})$, which follows by simultaneous diagonalization of the commuting projections \(P_{j,n}\),
we obtain
\bb
1-
\Tr\!\left[
\rho_n\prod_{j=1}^m P_{j,n}
\right]
\le
\sum_{j=1}^m
\Tr[\rho_n(\id_n-P_{j,n})]
\longrightarrow0 \quad \text{as} \quad n \to \infty.
\ee
\end{proof}
\medskip
{Integrable quantum many-body systems possess, in addition to the Hamiltonian, an extensive family of mutually commuting conserved quantities
\(
Q_1,Q_2,\ldots
\)
that strongly constrain their dynamics, preventing conventional thermalization to the Gibbs state\footnote{By contrast, for generic non-integrable systems, equilibrium behaviour is often expected to be governed by the Eigenstate Thermalization Hypothesis (ETH) (see e.g.~\cite{Srednicki_1994, Deutsch1991}), according to which local observables thermalize to ordinary Gibbs states determined solely by the energy density.}. Typical examples include quantum spin chains such as the XXZ model, whose conserved quantities contain the Hamiltonian, total magnetization, and higher-order multi-site observables. As a consequence of these additional conservation laws, such systems generally equilibrate not to ordinary Gibbs states, but rather to generalized Gibbs ensembles (GGEs) (see e.g.~\cite{VidmarRigol2016}) of the form
\bb
\gamma_{\boldsymbol\lambda}
=
\frac{
\exp\!\left(
-\lambda_0 H-\sum_j \lambda_j Q_j
\right)
}{
\Tr\!\left[
\exp\!\left(
-\lambda_0 H-\sum_j \lambda_j Q_j
\right)
\right]
},
\ee
where \(\boldsymbol\lambda=(\lambda_0,\lambda_1,\ldots,\lambda_m)\) is a collection of fixed real parameters\footnote{The parameters \(\lambda_j\) arise as Lagrange multipliers enforcing fixed expectation values of the conserved quantities \(H,Q_1,\ldots,Q_m\) in the maximum-entropy principle underlying generalized Gibbs ensembles. In particular, \(\lambda_0\) plays the role of the inverse temperature associated with the Hamiltonian \(H\), while the remaining parameters may be interpreted as generalized chemical potentials corresponding to the additional conserved quantities.}.

{We now specialize to mutually commuting self-adjoint single-site
operators $$H,Q_1,\ldots,Q_m \in \mathcal L(\mathcal H).$$
The generalized Gibbs state defined above is therefore a single-site
state $\gamma_{\bm\lambda}\in\mathcal D(\mathcal H)$.
We consider quantum sources $\bm\rho=(\rho_n)_n$ that are weakly
almost iid along this fixed state. The following corollary shows
that their empirical observables exhibit the same concentration
behaviour as those of the iid tensor-power source
$(\gamma_{\bm\lambda}^{\otimes n})_n$.}

\begin{cor}[(GGE typicality)]
\label{cor:gge-typicality}
Let $\bm\rho=(\rho_n)_n$ be weakly almost iid along $\gamma_{\boldsymbol\lambda}$. Then, for every $\delta>0$, the empirical observables
$$\overline H_n:=\frac{1}{n}\sum_i H^{(i)} \quad \text{and} \quad
\overline Q_{\ell,n}:=\frac{1}{n}\sum_i Q_\ell^{(i)}$$
simultaneously concentrate around their GGE expectation values:
\bb
\Tr\!\left[
\rho_n
\left\{
|\overline H_n-\Tr(\gamma_{\boldsymbol\lambda}H)\id_n|\le\delta
\right\}
\prod_{\ell=1}^m
\left\{
|\overline Q_{\ell,n}-\Tr(\gamma_{\boldsymbol\lambda}Q_\ell)\id_n|\le\delta
\right\}
\right]
\to1 \quad \text{as} \quad n \to \infty.
\ee
\end{cor}

\begin{proof}
Apply Theorem~\ref{thm:thermodynamic-typicality} with reference state $\rho=\gamma_{\boldsymbol\lambda}$ and commuting observables $H,Q_1,\ldots,$
$ \ldots, Q_m$.
\end{proof}

Thus, even in the presence of arbitrary long-range correlations and multipartite
entanglement, weakly almost iid states along a GGE reproduce the same
macroscopic equilibrium behaviour as the corresponding tensor-power GGE state.
{
\subsection{Measurement statistics of weakly almost iid sources}\label{sec:mmt}

{We next show that weakly almost iid structure is reflected directly in the statistics of local measurements. In particular, empirical outcome frequencies obtained from repeated local measurements concentrate around the probabilities predicted by the reference state, just as in the genuinely iid setting.}

\begin{thm}[(Concentration of empirical measurement frequencies)]
\label{thm:measurement-statistics-concentration}
Let $\rho\in\mathcal D(\mathcal H)$, with $\dim\mathcal H<\infty$, and let
$\bm\rho=(\rho_n)_n$, with
$\rho_n\in\mathcal D(\mathcal H^{\otimes n})$, be weakly almost iid along $\rho$.
Let $\mathsf M=\{M_x\}_{x\in\mathcal X}$ be a POVM on $\mathcal H$, where $\mathcal X$ is finite, and define
$p(x):=\Tr(\rho M_x)$.
Measuring \(\mathsf M^{\otimes n}\) on \(\rho_n\) induces a probability distribution\(p_n\) on \(\mathcal X^n\). Let \(\mathbb P_n\) denote probabilities with respect to \(p_n\). If \((X_1,\ldots,X_n)\sim p_n\), define the empirical frequencies
\bb
\hat p_n(x):=\frac1n\sum_{i=1}^n \mathbf 1_{\{X_i=x\}},\qquad x\in\mathcal X.
\ee
Then, for every \(\delta>0\),\bb{\mathbb P}_n\!\left(\max_{x\in\mathcal X}|\hat p_n(x)-p(x)|>\delta\right)\to0\quad\text{as}\quad n\to\infty.\ee

\end{thm}

\begin{remark}
Thus, although weakly almost iid sources may exhibit strong global
correlations, the empirical outcome frequencies produced by repeated
local measurements with the POVM $\mathsf M$ obey the same asymptotic
law of large numbers as independent samples: they converge in
probability to the probabilities $p(x)=\operatorname{Tr}(\rho M_x)$
predicted by the reference state. Moreover, weakly almost iid
structure is preserved under any fixed local measurement.
In particular, the empirical frequencies are consistent estimators
of the reference probabilities and, when $\mathsf M$ is
informationally complete, can be used to consistently estimate
the reference state $\rho$. This agreement of limiting values
does not imply asymptotic statistical indistinguishability of
the frequency distributions or agreement of their fluctuations.
\end{remark}

{\begin{proof}
Let $\mathcal M$ denote the CPTP measurement map associated with the POVM $M$\footnote{That is, for any $\omega \in \mathcal D(\mathcal H), \,\,$ 
\(
\mathcal M(\omega)
=
\sum_{x\in\mathcal X}
\Tr(\omega M_x)\, |x\rangle\!\langle x|.
\)
}. For every fixed \(k\) and every subset
\(I\subseteq[n]\) with \(|I|=k\), the marginal \((p_n)_I\) is obtained by applying \(\mathcal M^{\otimes k}\) to \((\rho_n)_I\). Hence, by contractivity of the trace norm,
\bb
\|(p_n)_I-p^{\otimes k}\|_1
\le
\|(\rho_n)_I-\rho^{\otimes k}\|_1 .
\ee
Averaging over all subsets \(I\subseteq[n]\) of size \(k\) and using the weakly almost iid property of \((\rho_n)_n\), it follows that \((p_n)_n\) satisfies the classical weakly almost iid condition along \(p\).

For each \(x\in\mathcal X\), define
\bb
f_x(y):=\mathbf 1_{\{y=x\}} .
\ee
Its expectation under \(p\) is \(p(x)\), and its empirical average under \(p_n\) is
\bb
\frac1n\sum_{i=1}^n f_x(X_i)
=
\hat p_n(x).
\ee
The same second-moment argument as in the proof of Lemma~\ref{lem:nc-weak-almost iid-lln} then gives
\bb
\mathbb E_{p_n}
\bigl[
(\hat p_n(x)-p(x))^2
\bigr]
\to0
 \quad \text{as} \quad n \to \infty.
\ee
Hence, by Chebyshev's inequality,
\bb
{\mathbb P}_n\!\left(
|\hat p_n(x)-p(x)|>\delta
\right)\to0
\quad \text{as} \quad n \to \infty
\ee
for every \(x\in\mathcal X\) and every \(\delta>0\). Since \(\mathcal X\) is finite, the union bound yields
\bb
{\mathbb P}_n\!\left(
\max_{x\in\mathcal X}
|\hat p_n(x)-p(x)|
>\delta
\right)
\le
\sum_{x\in\mathcal X}
{\mathbb P}_n\!\left(
|\hat p_n(x)-p(x)|>\delta
\right)
\to0 \,\, \text{as}\,\,n\to \infty.
\ee
\end{proof}}

\subsection{Smooth zero-Rényi entropy bounds}\label{sec:smooth}

Smooth entropy quantities~(see e.g.~\cite{Renner2005, Tomamichel2016, Konig_2009}) play a central role in one-shot quantum information theory,
where they characterize operational tasks such as data compression, randomness extraction,
and hypothesis testing in the non-asymptotic regime. In the asymptotic iid setting,
smooth entropies converge to the von Neumann entropy and related Shannon-theoretic
quantities. We now show that the universal entropy-concentration principle of
Lemma~\ref{lem:universal-typical-projector} directly yields an asymptotic upper bound on the smooth zero-Rényi entropy for weakly almost iid sources.

\begin{lemma}[Smooth zero-Rényi entropy bound]
\label{lem:smooth-zero-bound}
Let $\bm\rho=(\rho_n)_n$ be weakly almost iid along a state
$\rho\in\mathcal D(\mathcal H)$, where $\dim\mathcal H<\infty$.
For $\varepsilon\in(0,1)$, define
\bb
H_0^\varepsilon(\rho_n)
:=
\inf\left\{
\log(\Tr P):\ 
P \text{ is a projection and }
\Tr(P\rho_n)\ge 1-\varepsilon
\right\}.
\ee
Then
\bb
\limsup_{n\to\infty}
\frac1n
H_0^\varepsilon(\rho_n)
\le
S(\rho).
\ee
\end{lemma}

\begin{remark}
The quantity \(H_0^\varepsilon(\rho_n)\) measures the logarithmic dimension of the
smallest subspace capturing at least \(1-\varepsilon\) of the mass of \(\rho_n\). It therefore has a direct operational interpretation in one-shot quantum information theory, where it characterizes the size of approximate compression subspaces required for one-shot quantum compression of the state \(\rho_n\).
Lemma~\ref{lem:smooth-zero-bound} shows that every weakly almost iid source
asymptotically concentrates on subspaces whose exponential dimension is governed asymptotically by the entropy of the reference state \(\rho\),
just as in the genuinely iid setting.

For genuinely iid sources, the quantum asymptotic equipartition property (AEP)
establishes that suitably normalized smooth entropy quantities converge to the von
Neumann entropy in the asymptotic limit; see e.g.~\cite{Tomamichel_2009}. By contrast,
Lemma~\ref{lem:smooth-zero-bound} applies to the substantially broader class of
weakly almost iid sources, which may exhibit arbitrary long-range correlations and
entanglement. In this setting one cannot expect asymptotic equality in general:
globally pure weakly almost iid sources along the maximally mixed state provide
counterexamples. The lemma should therefore be viewed as an entropy-concentration
upper bound rather than a full asymptotic equipartition property.
\end{remark}

\begin{proof}
{Fix \(R>S(\rho)\). Since \(h_q\to S(\rho)\) as \(q\to0\),
we may choose \(q\in(0,1)\) and \(\delta>0\) such that
\(h_q+\delta<R\),
where
\bb
\rho_q:=(1-q)\rho+q\frac{\id}{d},
\qquad
h_q:=-\Tr(\rho\log\rho_q).
\ee}
For these choices of $q$ and $\delta$, Lemma~\ref{lem:universal-typical-projector}
provides projectors
\bb
\Pi_n
=
\left\{
-\frac1n\log\rho_q^{\otimes n}
\le
h_q+\delta
\right\}
\ee
satisfying
\bb
\Tr(\Pi_n\rho_n)\to1
\qquad\text{as}\qquad n\to\infty,
\ee
together with the rank bound
\bb
\Tr\Pi_n
\le
e^{n(h_q+\delta)}
<
e^{nR}
\ee
for every sufficiently large $n$.

Now fix $\varepsilon\in(0,1)$. Since
\(
\Tr(\Pi_n\rho_n)\to1,\) as $n \to \infty$,
there exists $n_0$ such that
\bb
\Tr(\Pi_n\rho_n)\ge1-\varepsilon
\qquad
\forall\, n\ge n_0 .
\ee
Hence, for all sufficiently large $n$, the projector $\Pi_n$ is admissible in the definition
of $H_0^\varepsilon(\rho_n)$. Therefore,
\bb
H_0^\varepsilon(\rho_n)
\le
\log\Tr\Pi_n
<
nR
\ee
for all sufficiently large $n$. Dividing by $n$ and taking the limit superior yields
\bb
\limsup_{n\to\infty}
\frac1n
H_0^\varepsilon(\rho_n)
\le
R.
\ee
Since $R>S(\rho)$ was arbitrary, the result follows.
\end{proof}

{\subsection{Spectral entropy rates and entropy concentration}\label{sec:info-spec}

The information-spectrum approach, initiated in the classical setting by Han and Verdú
and subsequently extended to quantum information theory by Hayashi and Nagaoka~\cite{HayashiNagaoka2003} (see also, e.g.~\cite{Hayashi2017} and references therein),
provides a powerful framework for studying information-theoretic tasks beyond the
iid setting. Rather than relying on asymptotic equipartition properties or tensor-product
structure, the information-spectrum method applies to completely general sequences of
states and channels, possibly exhibiting arbitrary correlations and non-stationary behaviour.

Within this framework, spectral entropy rates play a central role. In particular, the
spectral sup-entropy rate admits an operational interpretation as the optimal asymptotic
quantum compression rate for arbitrary quantum sources~\cite{BowenDatta2006}. Thus spectral entropy rates generalize the role
played by the von Neumann entropy in Schumacher compression to fully non-iid settings.

We now show that the universal concentration projectors of Lemma~\ref{lem:universal-typical-projector}
directly yield bounds on spectral entropy quantities for weakly almost iid sources.
The key observation is that the projectors produced by Lemma~\ref{lem:universal-typical-projector}
already satisfy the asymptotic concentration and dimension-growth conditions appearing
in the projector characterization of the spectral sup-entropy rate.}

We use the definition of the spectral sup-entropy rate introduced in~\cite{BowenDattaISIT2006}, which, as shown in that paper, is equivalent to the definition given in~\cite{HayashiNagaoka2003}. 

\begin{Def}[(Spectral sup-entropy rate)]
Let $\bm\omega=(\omega_n)_n$, with
$\omega_n\in\mathcal D(\mathcal H^{\otimes n})$.
For $\gamma\in\mathbb R$, define
\bb
A_n(\gamma)
:=
\omega_n-e^{-n\gamma}\id_n,
\ee
where \(\id_n\) denotes the identity operator on \(\mathcal H^{\otimes n}\).
The spectral sup-entropy rate of $\bm\omega$ is
\bb
\overline S(\bm\omega)
:=
\inf\left\{
\gamma:
\liminf_{n\to\infty}
\Tr\!\left[
\{A_n(\gamma)\ge0\}A_n(\gamma)
\right]
=1
\right\}.
\ee
\end{Def}

The spectral sup-entropy rate admits the following projector characterization. We include its proof for completeness.

\begin{prop}[(Projector characterization of the spectral sup-entropy rate)]
\label{prop:projector-characterization}
For every sequence $\bm\omega=(\omega_n)_n$,
\bb
\overline S(\bm\omega)
=
\inf\left\{
R:
\exists\ \text{projections }P_n
\text{ such that }
\Tr(P_n\omega_n)\to1 \text{\,as } n \to \infty\,;\,\,
\limsup_{n\to\infty}
\frac1n\log\Tr P_n
\le R
\right\}.
\ee
\end{prop}

\begin{proof}
Let \(R\) be such that
$$
\liminf_{n\to\infty}
\Tr[\{\omega_n\ge e^{-nR}\id_n\}
(\omega_n-e^{-nR}\id_n)]
=1.
$$
Define \(P_n:=\{\omega_n\ge e^{-nR}\id_n\}\).
{By the definition of \(\overline S(\bm\omega)\),
\bb
\liminf_{n\to\infty}
\Tr\!\left[
P_n(\omega_n-e^{-nR}\id_n)
\right]
=1.
\ee
Since
$
\Tr\!\left[
P_n(\omega_n-e^{-nR}\id_n)
\right]
\le
\Tr(P_n\omega_n)
\le 1
$
for every \(n\), it follows that
\bb
\Tr\!\left[
P_n(\omega_n-e^{-nR}\id_n)
\right]
\to1
\quad\text{as }\quad  n\to\infty.
\ee}
As
$
\Tr(P_n\omega_n)
\ge
\Tr\!\left[
P_n(\omega_n-e^{-nR}\id_n)
\right]
$
and \(\Tr(P_n\omega_n)\le1\), it follows that
\bb
\Tr(P_n\omega_n)\to1 \quad \text{as} \quad n \to \infty.
\ee
Moreover, since every eigenvalue of \(\omega_n\) on \(\Ran P_n\) is at least
\(e^{-nR}\),
\bb
1\ge \Tr(P_n\omega_n)\ge e^{-nR}\Tr P_n,
\ee
and hence \(\Tr P_n\le e^{nR}\) for all $n \in {\mathbb{N}}$.

Conversely, suppose that  \(P_n\) are projections satisfying
\bb
\Tr(P_n\omega_n)\to1 \,\, \text{as} \,\, n \to \infty, 
\quad\text{and}\quad
\limsup_{n\to\infty}\frac1n\log\Tr P_n\le R.
\ee
Fix \(\eta>0\). Then \(\Tr P_n\le e^{n(R+\eta)}\) for all sufficiently large \(n\). Define
\bb
Q_n:=\{\omega_n\ge e^{-n(R+2\eta)}\id_n\}.
\ee
Since \(Q_n\) is the positive spectral projection of
\(
\omega_n-e^{-n(R+2\eta)}\id_n,
\)
the Ky Fan maximum principle (see e.g.~\cite{BhatiaMatrixAnalysis}) gives
\bb
\Tr\!\left[
Q_n(\omega_n-e^{-n(R+2\eta)}\id_n)
\right]
\ge
\Tr\!\left[
P_n(\omega_n-e^{-n(R+2\eta)}\id_n)
\right].
\ee
Therefore,
\bb\label{eq:rhs}
\Tr\!\left[
Q_n(\omega_n-e^{-n(R+2\eta)}\id_n)
\right]
\ge
\Tr(P_n\omega_n)
-
e^{-n(R+2\eta)}\Tr P_n
\ge
\Tr(P_n\omega_n)-e^{-n\eta}.
\ee
The right-hand side of \eqref{eq:rhs} tends to \(1\) as $n \to \infty$. Since the left-hand side is bounded above by
\(\Tr(Q_n\omega_n)\le1\), it follows that
\bb
\Tr\!\left[
Q_n(\omega_n-e^{-n(R+2\eta)}\id_n)
\right]\to1 \quad \text{as} \quad n \to \infty .
\ee
By the definition of the spectral sup-entropy rate, this implies
\bb
\overline S(\bm\omega)\le R+2\eta.
\ee
Letting \(\eta\to0\) gives the claim.
\end{proof}

\begin{thm}
\label{thm:spectral-sup-entropy}
Let $\bm\rho=(\rho_n)_n$, with
$\rho_n\in\mathcal D(\mathcal H^{\otimes n})$, be weakly almost iid along
$\rho\in\mathcal D(\mathcal H)$, where $\dim\mathcal H<\infty$.
Then
\bb
\overline S(\bm\rho)\le S(\rho).
\ee
\end{thm}
\begin{remark}
{
 Thus, the universal entropy-concentration principle of Lemma~\ref{lem:universal-typical-projector} implies that the spectral sup-entropy rate of a weakly almost iid source is bounded above by the von Neumann entropy of the underlying reference state.}
\end{remark}
\begin{proof}
Fix $q\in(0,1)$ and $\delta>0$. By Lemma~\ref{lem:universal-typical-projector}, there exist projections $\Pi_n$ such that
$\Tr(\Pi_n\rho_n)\to1$ as $n \to \infty$, and
$$\limsup_{n\to\infty}\frac{1}{n}\log\Tr\Pi_n\le h_q+\delta,$$
where
$\rho_q:=(1-q)\rho+q\id/d$
and
$h_q:=-\Tr(\rho\log\rho_q)$.
Proposition~\ref{prop:projector-characterization} therefore yields
$\overline S(\bm\rho)\le h_q+\delta$.
Since $h_q\to S(\rho)$ as $q\to0$, and $\delta>0$ was arbitrary, the result follows. Importantly, the projectors entering the proof depend only on the reference state \(\rho\), and not on the individual source \((\rho_n)_n\).
\end{proof}

\section{Conclusions and open questions}

The results of this paper show that several characteristic features of the iid regime persist under the weakest of the currently studied forms of approximate iid structure. We identified two fundamental concentration principles for weakly almost iid quantum sources: a noncommutative weak law of large numbers for empirical observables, and a universal entropy-concentration principle yielding asymptotic concentration on subspaces whose exponential dimension is governed by the von Neumann entropy of the reference state. Together, these principles provide a unified and conceptually transparent framework for a range of information-theoretic applications beyond the tensor-power setting, including universal quantum data compression within classes of weakly almost iid sources sharing a common reference state, asymmetric quantum hypothesis testing, concentration results for quantum many-body systems and generalized Gibbs ensembles, concentration of repeated local measurement statistics, as well as bounds on smooth and spectral entropy quantities. In several cases, the resulting proofs are substantially more direct than the previous robustness-based approach of~\cite{girardi2026quantumshannontheoryrobust}, with the relevant operational statements emerging naturally from the underlying concentration principles.

A number of natural open questions remain. One important direction is to
extend the present framework to more general quantum information-theoretic
tasks, such as quantum communication, entanglement distillation, dense
coding, and decoupling-based protocols for weakly almost iid sources.
It would also be interesting to investigate possible second-order and
moderate-deviation refinements of the present results. On the many-body
side, an important problem is to extend the thermodynamic concentration
results beyond one-site observables to more general local or quasi-local
conserved quantities.

\subsection*{Acknowledgments}
The author is grateful to Filippo Girardi for introducing her to the notion of weakly almost i.i.d.~sources, and for many interesting discussions. She would also like to thank Bjarne Bergh for an insightful comment. ND is supported by the Engineering and Physical Sciences Research Council [Grant Ref: EP/Y028732/1]. 

\bibliography{biblio}

\section*{Appendix A: Existence of weakly almost iid pure-state sequences}\label{app:A}

For each \(n\in\mathbb N\), let
$
\rho_n=\ket{\psi_n}\!\bra{\psi_n},
$
where the vectors \(\ket{\psi_n}\in(\mathbb C^d)^{\otimes n}\) are chosen
according to the product Haar measure. Fix \(k\in\mathbb N_+\). For \(n\ge k\), define
\bb
X_{n,k}
:=
\mathbb E_{\substack{I\subseteq[n]\\ |I|=k}}
\left\|
(\rho_n)_I-\left(\frac{I}{d}\right)^{\otimes k}
\right\|_1 .
\ee

This is a nonnegative random variable.

From~\eqref{eq:haar-est},
\bb
\mathbb E_{\psi_n} X_{n,k}
\le
\sqrt{
d^k
\left(
\frac{d^k+d^{n-k}}{d^n+1}
-
\frac1{d^k}
\right)
}.
\ee
Moreover,
\bb
d^k
\left(
\frac{d^k+d^{n-k}}{d^n+1}
-
\frac1{d^k}
\right)
=
\frac{d^{2k}-1}{d^n+1}.
\ee
Hence
\bb
\mathbb E_{\psi_n} X_{n,k}
\le
\sqrt{
\frac{d^{2k}-1}{d^n+1}
}.
\ee
For fixed \(k\), there exists a constant \(C_k>0\) such that
\bb
\mathbb E_{\psi_n} X_{n,k}
\le
C_k d^{-n/2}.
\ee
Hence,
\bb
\sum_{n=k}^\infty
\mathbb E_{\psi_n} X_{n,k}
<
\infty.
\ee
Now fix \(\varepsilon>0\). By Markov's inequality,
\bb
\mathbb P(X_{n,k}>\varepsilon)
\le
\frac{\mathbb E_{\psi_n} X_{n,k}}{\varepsilon}.
\ee
Hence
\bb
\sum_{n=k}^\infty
\mathbb P(X_{n,k}>\varepsilon)
<
\infty.
\ee
By the first Borel--Cantelli lemma,
\bb
\mathbb P\bigl(
X_{n,k}>\varepsilon
\text{ infinitely often}
\bigr)=0.
\ee
Equivalently, with probability one, there exists
\(N_\varepsilon\in\mathbb N\) such that
\bb
X_{n,k}\le\varepsilon
\qquad
\forall n\ge N_\varepsilon.
\ee
Applying this to \(\varepsilon=1/m\), \(m\in\mathbb N\), and intersecting
the corresponding probability-one events, we conclude that
\bb
X_{n,k}\longrightarrow0
\quad {\hbox{as}} \quad
n\to\infty
\ee
almost surely.

Finally, since \(\mathbb N_+\) is countable, we may intersect the above
probability-one events over all \(k\in\mathbb N_+\). Therefore, with
probability one,
\bb
\mathbb E_{\substack{I\subseteq[n]\\ |I|=k}}
\left\|
(\rho_n)_I-\left(\frac{I}{d}\right)^{\otimes k}
\right\|_1
\longrightarrow0 \quad {\hbox{as}} \quad
n\to\infty
\ee
for every fixed \(k\in\mathbb N_+\).
\end{document}

%% file: macros.tex
\usepackage{newpxtext,newpxmath}

\usepackage[utf8]{inputenc}
\usepackage{amsthm}
\usepackage{amssymb}
\usepackage{amsmath}
\usepackage{bbold}
\usepackage{bbm}
\usepackage[pdftex, backref=page]{hyperref}
\usepackage{braket}
\usepackage{dsfont}
\usepackage{mathdots}
\usepackage{mathtools}
\usepackage{enumerate}
\usepackage[shortlabels]{enumitem}
\usepackage{csquotes}
\usepackage{stmaryrd}
\usepackage[cal=boondox]{mathalfa}
\usepackage{graphicx}
\usepackage{stackengine}
\usepackage{scalerel}
\usepackage{tensor}       
\usepackage{array}
\usepackage{makecell}
\newcolumntype{x}[1]{>{\centering\arraybackslash}p{#1}}
\usepackage{tikz}
\usepackage{pgfplots}
\usetikzlibrary{shapes.geometric, shapes.misc, positioning, arrows, arrows.meta, decorations.pathreplacing, decorations.pathmorphing, patterns, angles, quotes, calc}
\usepackage{booktabs}
\usepackage{xfrac}
\usepackage{siunitx}
\usepackage{centernot}
\usepackage{comment}
\usepackage{chngcntr}
\usepackage{caption}
\usepackage{subcaption}

\newtheorem{thm}{Theorem}
\newtheorem*{thm*}{Theorem}
\newtheorem{prop}[thm]{Proposition}
\newtheorem*{prop*}{Proposition}
\newtheorem{lemma}[thm]{Lemma}
\newtheorem*{lemma*}{Lemma}
\newtheorem{cor}[thm]{Corollary}
\newtheorem*{cor*}{Corollary}

\newtheorem*{cj*}{Conjecture}
\newtheorem{Def}[thm]{Definition}
\newtheorem*{Def*}{Definition}

\newtheorem*{question*}{Question}

\newtheorem*{problem*}{Problem}
\newtheorem*{example*}{Example}

\makeatletter
\def\thmhead@plain#1#2#3{%
  \thmname{#1}\thmnumber{\@ifnotempty{#1}{ }\@upn{#2}}%
  \thmnote{ {\the\thm@notefont#3}}}
\let\thmhead\thmhead@plain
\makeatother

\theoremstyle{definition}

\newcommand{\bb}{\begin{equation}\begin{aligned}\hspace{0pt}}
\newcommand{\bbb}{\begin{equation*}\begin{aligned}}
\newcommand{\ee}{\end{aligned}\end{equation}}
\newcommand{\eee}{\end{aligned}\end{equation*}}

\newcommand{\ketbra}[1]{\ket{#1}\!\!\bra{#1}}

\renewcommand{\epsilon}{\varepsilon}

\newcommand{\id}{\mathds{1}}

\DeclareMathOperator{\Tr}{Tr}

\DeclareMathAlphabet{\pazocal}{OMS}{zplm}{m}{n}

\DeclareMathOperator{\supp}{supp}

\newcommand{\lsmatrix}{\left(\begin{smallmatrix}}
\newcommand{\rsmatrix}{\end{smallmatrix}\right)}

\stackMath

\stackMath

\makeatletter
\newcommand*\rel@kern[1]{\kern#1\dimexpr\macc@kerna}
\newcommand*\widebar[1]{%
  \begingroup
  \def\mathaccent##1##2{%
    \rel@kern{0.8}%
    \overline{\rel@kern{-0.8}\macc@nucleus\rel@kern{0.2}}%
    \rel@kern{-0.2}%
  }%
  \macc@depth\@ne
  \let\math@bgroup\@empty \let\math@egroup\macc@set@skewchar
  \mathsurround\z@ \frozen@everymath{\mathgroup\macc@group\relax}%
  \macc@set@skewchar\relax
  \let\mathaccentV\macc@nested@a
  \macc@nested@a\relax111{#1}%
  \endgroup
}

\counterwithin*{equation}{part}
\counterwithin*{thm}{part}
\counterwithin*{figure}{part}

\tikzset{meter/.append style={draw, inner sep=10, rectangle, font=\vphantom{A}, minimum width=30, line width=.8, path picture={\draw[black] ([shift={(.1,.3)}]path picture bounding box.south west) to[bend left=50] ([shift={(-.1,.3)}]path picture bounding box.south east);\draw[black,-latex] ([shift={(0,.1)}]path picture bounding box.south) -- ([shift={(.3,-.1)}]path picture bounding box.north);}}}
\tikzset{roundnode/.append style={circle, draw=black, fill=gray!20, thick, minimum size=10mm}}
\tikzset{squarenode/.style={rectangle, draw=black, fill=none, thick, minimum size=10mm}}

\definecolor{Blues5seq1}{RGB}{239,243,255}
\definecolor{Blues5seq2}{RGB}{189,215,231}
\definecolor{Blues5seq3}{RGB}{107,174,214}
\definecolor{Blues5seq4}{RGB}{49,130,189}
\definecolor{Blues5seq5}{RGB}{8,81,156}

\definecolor{Greens5seq1}{RGB}{237,248,233}
\definecolor{Greens5seq2}{RGB}{186,228,179}
\definecolor{Greens5seq3}{RGB}{116,196,118}
\definecolor{Greens5seq4}{RGB}{49,163,84}
\definecolor{Greens5seq5}{RGB}{0,109,44}

\definecolor{Reds5seq1}{RGB}{254,229,217}
\definecolor{Reds5seq2}{RGB}{252,174,145}
\definecolor{Reds5seq3}{RGB}{251,106,74}
\definecolor{Reds5seq4}{RGB}{222,45,38}
\definecolor{Reds5seq5}{RGB}{165,15,21}

\allowdisplaybreaks

\usepackage[most,breakable]{tcolorbox}
	{\expandafter\ifstrequal\expandafter{#1}{orange}{\begin{tcolorbox}[colback=red!15,colframe=orange!15,breakable,enhanced]}{\begin{tcolorbox}[colback=Blues5seq1,colframe=Blues5seq5,breakable,enhanced]}}%
	{\end{tcolorbox}}